\documentclass[10pt,twocolumn]{IEEEtran}
\IEEEoverridecommandlockouts

\usepackage{amsthm, amssymb}

\def\b{\mathbb}
\def\bb{\boldsymbol}

\def\sinr{\textrm{SINR}}

\def\dint{{\rm \ d}}

\newtheoremstyle{slplain}
  {3pt}
  {3pt}
  {\slshape}
  {}
  {\bfseries}
  {.}%
  { }
  {}

\theoremstyle{slplain}

\newtheorem{cor}{Corollary}
\newtheorem{lem}{Lemma}
\newtheorem{pro}{Proposition}

\newtheorem{assum}{Assumption}

\usepackage{cite}
\usepackage{amsfonts}
\usepackage{multicol,multienum}
\usepackage{subfigure}
\usepackage{multirow}

\ifCLASSINFOpdf
  \usepackage[pdftex]{graphicx}
  \graphicspath{{../pdf/}{../jpeg/}}
  \DeclareGraphicsExtensions{.pdf,.jpeg,.png}
\else
  \usepackage{float}
  \usepackage[dvips]{graphicx}
  \graphicspath{{../eps/}}
  \DeclareGraphicsExtensions{.eps}
\fi

\usepackage[cmex10]{amsmath}
\usepackage{algorithm}
\usepackage{algorithmic}
\usepackage{array}
\usepackage{url}
\usepackage{setspace}

\begin{document}

\title{Modeling, Analysis and Design for Carrier Aggregation in Heterogeneous Cellular Networks}

\author{
\IEEEauthorblockA{Xingqin Lin, Jeffrey G. Andrews and Amitava Ghosh}
\thanks{Xingqin Lin and Jeffrey G. Andrews are with Department of Electrical $\&$ Computer Engineering, The University of Texas at Austin, USA. (E-mail: xlin@utexas.edu, jandrews@ece.utexas.edu).  Amitava Ghosh is with Nokia Siemens Networks. (E-mail: amitava.ghosh@nsn.com).  This research was supported by Nokia Siemens Networks.}
}

\maketitle

\begin{abstract}
Carrier aggregation (CA) and small cells are two distinct features of next-generation cellular networks. Cellular networks with different types of small cells are often referred to as HetNets. In this paper, we introduce a load-aware model for CA-enabled \textit{multi}-band HetNets.  Under this model, the impact of biasing can be more appropriately characterized; for example, it is observed that with large enough biasing, the spectral efficiency of small cells may increase while its counterpart in a fully-loaded model always decreases. Further, our analysis reveals that the peak data rate does not depend on the base station density and transmit powers; this strongly motivates other approaches e.g. CA to increase the peak data rate. Last but not least, different band deployment configurations are studied and compared. We find that with large enough small cell density, spatial reuse with small cells outperforms adding more spectrum for increasing user rate. More generally, universal cochannel deployment typically yields the largest rate; and thus a capacity loss exists in orthogonal deployment. This performance gap can be reduced by appropriately tuning the HetNet coverage distribution (e.g. by optimizing biasing factors). 
\end{abstract}

\IEEEpeerreviewmaketitle

\section{Introduction}

Carrier aggregation (CA) is considered as a key enabler for LTE-Advanced \cite{3gppLTER10}, which can meet or even exceed the IMT-Advanced requirement for large transmission bandwidth (40 MHz-100 MHz) and high peak data rate (500 Mbps in the uplink and 1 Gbps in the downlink). Simply speaking, CA enables the concurrent 
utilization of multiple component carriers on the physical layer to expand the effective bandwidth \cite{Iamura2010CA, Ghosh2010CA, yuan2010carrier}. The aggregated bandwidth can be as large
as 100 MHz, for example, by aggregating 5 component carriers of bandwidth 20 MHz each. The bandwidth of these component carriers can vary widely (ranging from 1.4 MHz to 20 MHz for LTE carriers \cite{3gppPhy, Parkvall2008LTE, lin2013dynamic}). The propagation characteristics
of different component carriers may vary significantly, e.g., a component carrier in the 800 MHz band has very different propagation characteristic from a component carrier in the 2.5 GHz band. Due to the significance and unique features of CA, appropriate CA management is essential for enhancing the performance of CA-enabled cellular networks.

\subsection{Related Work and Motivation}

Cellular networks are undergoing a major evolution as current cellular networks cannot keep pace with user demand through simply deploying more macro base stations (BS) \cite{AndCla12}. As a result, attention is being shifted to deploying small, inexpensive, low-power nodes in the current macro cells; these low power nodes may include  pico \cite{landstrom2011deployment} and femto \cite{AndCla12} BSs,  as well as distributed antennas \cite{clark2001distributed}. Cellular networks with them take on a very heterogeneous characteristic, and are often referred to as HetNets \cite{qual2011LTE, eric2011het,  damnjanovic2011survey, andrews2013seven}. Due to the heterogeneous and ad hoc deployments common in low power nodes, the validity of adoption of the classical models such as Wyner model \cite{wyner1994shannon} or hexagonal grid \cite{3gppPhy} for HetNet study becomes questionable \cite{andrews2010tractable}.

Not surprisingly, random spatial point processes \cite{stoyan1995stochastic}, particularly homogeneous Poisson point process (PPP) for its tractability, have been used to model the locations of the various types of heterogeneous BSs. Such a probabilistic approach for cellular networks can be traced back to late 1990's \cite{baccelli1997stochastic, brown2000cellular}. The PPP does not exactly capture every characteristic of cellular networks; for example, two points in a PPP can infrequently be arbitrarily close to each other, which is not usually true in practice, especially for the deployment of macro BSs. Nevertheless, as a null hypothesis, the PPP model opens up tractable ways of assessing the statistical properties of cellular networks, and it provides reasonably good performance prediction. Indeed, the recent work \cite{andrews2010tractable, lin2012towards} have demonstrated that the PPP model is about as accurate in terms of downlink SINR distribution and handover rate as the hexagonal grid for a representative urban/suburban cellular network; the gap is even smaller for the uplink SINR distribution with channel inversion \cite{lin2013optimal}. Further, the PPP model may become even more accurate in HetNets due to the heterogeneous and ad hoc deployments common in low power nodes\cite{andrews2010tractable}.

Using the PPP model, the locations of different type of BSs in a HetNet are often modeled by an independent PPP \cite{dhillon2011modeling}. Due to the tractability of PPP model, many analytical results such as coverage probability and rate can be obtained \cite{dhillon2011modeling, Sayandev2012dis, jo2011tractable}; more interestingly, these analytical results fairly agree with industry findings obtained by extensive simulations and experiments \cite{nsn2012het}. As a result, similar models have been further used to optimize the HetNet design including spectrum allocation \cite{Cheung2012throughput}, load balancing \cite{singh2012offloading, Sung2013Energy}, spectrum sensing \cite{Hesham2013two}, etc. These encouraging progresses motivate us to adopt the PPP model for CA study in this paper.

One major concern about deploying small cells is that they have limited coverage due to their low transmit powers. As a result, small cells are often lightly loaded and would not accomplish much without load balancing, while macro cells are still heavily loaded. To alleviate this issue, a simple load balancing approach called biasing has been proposed \cite{qual2011LTE}; biasing allows the low power nodes to \textit{artificially} increase their transmit powers.  As a result, it helps expand the coverage areas of small cells and enables more user equipment (UE) to be served by small cells \cite{nsn2012het, qual2011LTE, lopez2012expanded}. It is expected that biasing can help balance network load  and correspondingly leads to higher throughput. However, a theoretical study on the impact of biasing is challenging. While \cite{jo2011tractable} modeled biasing in a \textit{single}-band HetNet, it assumes a fully-loaded HetNet, i.e., all the BSs are simultaneously active all the time. Similar fully-loaded model is also used in \cite{singh2012offloading} for offloading study.

Sum rate is an important metric in wireless networks, particularly CA-enabled HetNets. Unfortunately, from the sum-rate perspective, biasing does not help in a fully-loaded network. Thus, in order to appropriately examine the effect of biasing on the sum rate of CA-enabled HetNets, an appropriate notion of BS load is needed. While \cite{dhillon2012load} proposed a load-aware model in which low power nodes can be less active than macro BSs over the \textit{time} domain, it focuses on \textit{single}-band HetNet and is still not sufficient for CA study which essentially involves multi-band modeling and analysis. Therefore, the main goal of this paper is to propose a multi-band HetNet model with an appropriate notion of load. With the proposed model, we would like to theoretically examine the impact of biasing and study how to deploy the available bands in a HetNet to best exploit CA.

\subsection{Contributions and Main Results}

The main contributions and outcomes of this paper are as follows.

\subsubsection{A new load-aware multi-band HetNet model}

In Section \ref{sec:system}, we introduce a load-aware model for CA-enabled \textit{multi}-band HetNets. Compared to \cite{dhillon2012load} which uses a time-domain notion of load, the new model uses a notion of fractional load in the frequency domain. While the former is suitable for bursty traffic, the latter is more suitable for static traffic like VoIP. Moreover, the latter provides a different view on load modeling in HetNets and thus can be viewed as complementary to the former. The proposed model is flexible enough to capture the main unique features of both multi- and single flow CA (defined in Section \ref{sec:system}) but yet tractable enough for analysis. Under this load-aware model, the impact of biasing is appropriately characterized; for example, it is observed that with large enough biasing, the spectral efficiency of small cells can actually increase while its counterpart in a fully-loaded model always decreases because of the signal-to-interference-plus-noise ratio (SINR) reduction from biasing.

\subsubsection{Rate analysis in multi-band HetNets}

Unlike rate analysis in a single-band HetNet, there exists correlation among the signals and interference across the bands in a multi-band HetNet. This correlation feature is highlighted in Section \ref{sec:multiband}. In this paper, we present a way to deal with this correlation first for single tier networks (c.f. Section \ref{sec:multiband}) and then extend it for general multi-tier HetNets (c.f. Section \ref{sec:single}). This way of dealing with the multi-band correlation contributes to the tractable approach to the performance analysis of cellular networks \cite{andrews2010tractable}.

\subsubsection{Design insights}

From the analytical results, several observations may be informative for system design. In single tier interference-limited networks (where noise is ignored), it is found that the peak data rate does not depend on the BS density and transmit powers; this strongly motivates other approaches e.g. CA to increase the peak data rate. Further, if the aggregated carriers are sorted in ascending order based on their path-loss exponents, the peak data rate scales \textit{super-linearly} with the number of aggregated carriers; this  provides a finer characterization for the common conjecture of \textit{linearly} scaled peak data rate with increasing number of carriers\cite{Ghosh2010CA}.

Different band deployment configurations are studied and compared for CA-enabled HetNets. We derive the rate expressions of $1$-band-$K$-tier deployment and $K$-band-$1$-tier deployment; these two deployments represent two popular approaches for increasing rate in cellular networks: spatial reuse with small cells and adding more bandwidth. It is found that, if the densities of low power nodes are large enough, a $1$-band-$K$-tier deployment can provide larger rate than the $K$-band-$1$-tier deployment. This gives additional theoretical justification for using small cells to solve the current ``spectrum crunch''. More generally, we find that universal cochannel deployment -- all the tiers use all the bands -- typically yields the largest rate. Correspondingly, there is a capacity loss in orthogonal deployment -- different tiers use different bands. This performance gap can be reduced by appropriately tuning the HetNet coverage distribution (e.g. by optimizing biasing factors).

\section{System Model}
\label{sec:system}

We consider a general HetNet with $\mathcal{K} = \{1,...,K\}$ denoting the set of $K$ tiers which may include macro cells, pico cells, femtocells, and possibly other elements. BSs of different types may differ in terms of deployment density, spectrum resource, transmit power and supported modulation and coding scheme. In this paper, we focus on the downlink and assume open access for all the small cells. The other key aspects of the studied model are described below.

\subsection{Distributions of BSs and UEs}

The BS locations are modeled as $K$ independent homogeneous PPPs \cite{stoyan1995stochastic}. Denote by $\Phi_k$ the set of PPP distributed BSs in tier $k$ and $\lambda_k$ its density. The UEs are also assumed to be randomly distributed according to an independent PPP of density $\lambda^{(u)}$. Or equivalently, the number of UEs in a certain region follows Poisson distribution whose mean equals $\lambda^{(u)}$ multiplied by the area of the region, and given the number, UEs are independently and uniformly located in the region. Uniform random UE spatial 
distribution over a certain region is often utilized by industry in system level simulations (see e.g. \cite{3gppOfdm}).

\subsection{Channel Model}

We assume that there are a set of $M$ available \textit{bands} denoted as  $\mathcal{M}= \{1, 2, ..., M\}$. The bandwidth and path loss exponent of each band $i$ are denoted by $B_i$ and $\alpha_i$, respectively. Here we use different path loss exponents for different bands to capture the possibly large differences in propagation characteristics associated with each band's carrier frequency. We further assume that each band $i$ is small enough to have relatively constant path loss exponent $\alpha_i$ across it.

Suppose each tier-$k$ BS transmits at constant power $P_{i,k}$ in the $i$-th band, provided that band $i$ is used by tier $k$. Then the received power $\tilde{P}_{i,k}$ at the typical UE (assumed to be at the origin) from the BS located at $Y \in \b R^2$ is modeled as
\begin{align}
\tilde{P}_{i,k, Y} = P_{i,k} H_{i,Y} C_i \| Y \|^{ - \alpha_i},
\end{align}
where $\| Y \|$ denotes the Euclidean norm of $Y$, $C_i$ is a constant that gives the path loss in band $i$ when the link length is $1$, and $H_{i,Y}$ is a random variable capturing the fading value of the radio link from the BS at $Y$ to the typical UE in band $i$.  Note that $C_i$ strongly  depends on carrier frequency, e.g. $C_i \cong (\mu_i/4 \pi)^2 $ where $\mu_i$ denotes the wavelength. For simplicity, we ignore shadowing and consider Rayleigh fading only, i.e., $H_{i,Y} \sim \textrm{Exp}(1)$. In fact, the randomness of the BS locations actually helps to emulate shadowing: As shadowing variance increases, the resulting propagation losses between the BSs and the typical user in a grid network converge to those in a Poisson distributed network \cite{Blaszczyszyn2012using}.

\subsection{User Association Schemes}
\label{subsec:user}

We introduce two types of CA, as shown in Fig. \ref{fig:1}. The first type is called \textit{multi-flow} CA, where UEs can be potentially associated with all the available tiers simultaneously (but in different bands) and can aggregate data using all the available bands. Then the typical UE is associated with tier $k$ in band $i$  that provides the maximum biased received power, i.e., $k = \arg \max (Z_\ell P_{i,\ell } \| Y_{\ell_0} \|^{ - \alpha_i } : \ell \in \mathcal{K}  ) $,
where $Z_\ell$ denotes the biasing factor of tier $\ell$ and $Y_{\ell_0}$ denotes the location of the nearest BS in tier $\ell$. Biasing factors $\boldsymbol{Z}=\{Z_k:k\in \mathcal{K}\}$ are used in HetNets for the purpose of load balancing: Adopting larger $Z_k$ expands the cell range of BSs in tier $k$ and thus more UEs can connect to tier $k$. 

\begin{figure}
\centering
\includegraphics[width=8cm]{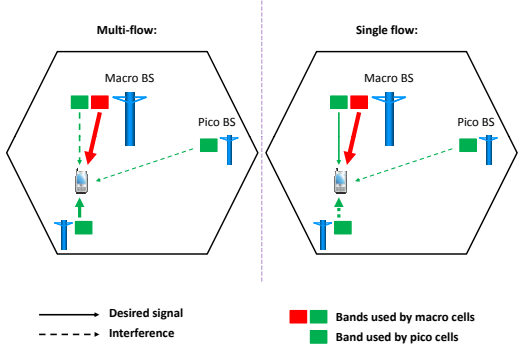}
\caption{System model: multi-flow versus single flow}
\label{fig:1}
\end{figure}

The second type is called \textit{single flow} CA, where UEs can be associated with only one of the available tiers at a time, i.e., one BS at some tier, though they still can aggregate data using all the available bands used by that tier. For this type of CA, the typical UE scans over all the tiers and bands and  connects to the tier that provides the strongest biased received power in some band, say $i^*$. Then the UE performs CA with respect to this tier. Formally, the UE connects to tier $k$ such that $ (k,i^*) = \arg \max (Z_\ell P_{i,\ell} \| Y_{\ell_0} \|^{ - \alpha_i } : (\ell,i) \in \mathcal{K} \times \mathcal{M}  ) $. Single flow may be closer to how CA is normally supported in reality: UE is configured with a primary component carrier that provides the best signal quality and other secondary component carriers are only added when applicable.

In summary, multi-flow allows the UE to perform CA across all the available bands in (possibly) different tiers, while single flow allows the UE to perform CA only across the available bands used by the selected tier. In this paper, we focus on single flow and refer to \cite{Lin2013CAICC} for the multi-flow study.


\subsection{Load Modeling}

We assume that each UE connecting to tier $k$ in band $i$ requires a basic share $b_{i,k}$ of the bandwidth resource $B_i$. For example,  $b_{i,k}$ may represent the basic frequency resource allocation unit in LTE networks. This requirement can be satisfied if tier $k$ in band $i$ is under-loaded, e.g., $L_{i,k} \leq B_i/b_{i,k}$ where $L_{i,k}$ denotes the mean number of UEs attempting to connect to a tier-$k$ BS in band $i$. If $L_{i,k} > B_i/b_{i,k}$, this implies that too many UEs attempt to connect to tier $k$ and correspondingly tier $k$ becomes fully-loaded. In this case, some UEs within the coverage of tier $k$ have to be blocked. Or equivalently, assuming all the UEs are admitted by tier $k$, they can only obtain a fraction of the time domain resource.\footnote{This is essentially a round-robin scheduling mechanism. Other more sophisticated scheduling mechanisms (e.g. proportional fair scheduling) may be considered. However, such extensions may significantly complicate the analysis and we treat them as future work. Nevertheless, our current analysis can provide a lower bound on the performance of the more sophisticated scheduling mechanisms.}

In the following, we stick to the former interpretation and will derive the admission probability while keeping in mind that the admission probability can also be interpreted as the fraction of the time domain resource shared by the UEs with no blocking. For each BS in the under-loaded tier $k$, we assume that  the location of the used bandwidth $b_{i,k}L_{i,k}$ in band $i$ is uniformly and independently selected from band $i$. Equivalently, frequency hopping can be assumed. We further define 
\begin{align}
\theta_{i,k} = \min \left(\frac{b_{i,k}L_{i,k}}{B_i},1 \right), \forall k \in \mathcal{K}, \forall i \in \mathcal{M},
\end{align}
and refer to $\theta_{i,k}$ as the \textit{load} of tier $k$ in band $i$. Accordingly, the spectrum resources of band $i$ currently allocated by tier $k$ equals $\theta_{i,k} B_i$. 

This bandwidth sharing model helps capture the fact that the load in HetNets is often unbalanced. In particular, small BSs have limited coverage and thus the mean number $L_{i,k}$ of UEs served by per small BS  is small, yielding light load $\theta_{i,k} \ll 1$. In contrast, macro BSs have much larger coverage and thus the mean number $L_{i,k}$ of UEs served by per macro BS is large, yielding heavy load $\theta_{i,k}\approx 1$. Further, the bandwidth sharing model is flexible enough to model different network scenarios: the larger $b_{i,k}$ is, the more data-hungry UEs exist in the HetNet; in an extreme case, we can set $b_{i,k}=B_i$, which makes all the available spectrum resources at a BS be consumed as long as the BS serves at least one UE.

\subsection{Performance Metric}

The performance of the HetNet depends on how the available bands are deployed. To model band deployment, we introduce binary decision variables defined as
\[ x_{i,k} = \left\{ \begin{array}{ll}
         1 & \mbox{if band $i$ is used by tier $k$};\\
        0 & \mbox{otherwise}, \end{array} \right. \] 
for all $i \in \mathcal{M}, k \in \mathcal{K}$. We group $x_{i,k}$'s into a vector $\bb x$ denoting the band deployment configuration and make the following assumption on band deployment.
\begin{assum}
Each band is used by at least one tier and each tier uses at least one band, i.e.,
$\sum_{ k \in \mathcal{K} } x_{i,k} \geq 1, \forall i$, and $\sum_{ i \in \mathcal{M} } x_{i,k} \geq 1, \forall k$.
\label{asum:1}
\end{assum}
There is no loss of generality in the above assumption since one can simply exclude the unused band and/or the useless tier from consideration. Given a band deployment configuration, the signal-to-interference-plus-noise ratio (SINR) in band $i$ provided by tier $k$ can be computed as
\begin{align}
\textrm{SINR}_{i,k} = \frac{ x_{i,k} P_{i,k} H_{i,k_0} \parallel Y_{k_0} \parallel^{-\alpha_i}   }{ \sum_{\ell \in \mathcal{K}}  \sum_{Y \in \tilde{\Phi}_{i,\ell} \backslash Y_{k_0} } P_{i,l} H_{i,Y}  \parallel Y \parallel^{-\alpha_i} + w_{i,k}  }, \notag
\end{align}
where $w_{i,k} = b_{i,k} N_0 / C_i$ with $N_0$ being the power spectral density of the background noise, and $\tilde{\Phi}_{i,k}$ is the PPP of density $\theta_{i,k} \lambda_k$ thinned from $\Phi_k$.

Rate, a function of the received $\sinr$, is the paramount metric in CA which aims to provide UEs with very high data rate. This motivates us to adopt the \textit{UE ergodic rate} as the metric for CA study in this paper. UE ergodic rate measures the long term data rate attained by a typical UE and will be derived in the following sections.

Before ending this section, we would like to stress that though not modeled in this paper, interference management is also a critical part of HetNets, especially for cochannel deployment with biasing; otherwise, aggressive biasing may result in unacceptably poor SINR performance of cell-edge users. We treat the incorporation of interference management as future work.

\section{Multi-Band Analysis: Single Tier Case}
\label{sec:multiband}

CA study essentially involves multi-band analysis, which is involved due to the spatial correlations in the multi-band signals and interference. To make this point explicit, let us consider a single tier cellular network that consists of macro BSs only and uses two bands, $1$ and $2$. In this case multi-flow and single flow CA coincide. Further, the general user association policies (c.f. Section \ref{subsec:user}) reduce to the simple nearest BS association: the typical UE connects to its nearest BS. Clearly,  the received signals in band $1$ and $2$ are both emitted from the nearest BS and thus are strongly correlated. By similar reasoning, the interference powers in band $1$ and $2$ are also correlated due to the presence of common randomness in the locations of the interferers.

In this section we analyze single tier cellular network to demonstrate how to cope with the spatial correlations in the multi-band analysis; we will extend it to general $K$ tier network later. As $K=1$, we shall drop the subscript $k$ in this section for ease of notation. 


To begin with, we know that there are $\lambda^{(u)}/\lambda$ UEs\footnote{This can be shown rigorously by using Neveu's exchange formula (see e.g. \cite{baccelli2009stochastic}).} on average associated with a typical BS and correspondingly $L_i = \lambda^{(u)}/\lambda, \forall i$. Then by definition the load of each band $i$ is given by
\begin{align}
\theta_i = \min \left( \frac{b_{i} L_i }{B_i} , 1 \right) = \min \left( \frac{\lambda^{(u)} b_{i}}{\lambda B_i} , 1 \right).
\end{align} 
It follows that the effective transmitter process $\tilde{\Phi}_i$ in band $i$ is a PPP with density $\theta_i \lambda$, thinned from the common ground transmitter process $\Phi$ of density $\lambda$. Clearly, $\{\tilde{\Phi}_i\}$ are not independent and thus spatial correlations are induced.

Now let us condition on the event that the typical UE connects to the BS $Y_{0}$ located at a distance $r$ from the UE. Here comes a tricky thing: conditioned on connecting to the BS $Y_{0}$, this tagged BS has a coverage area containing the typical UE and thus statistically it has a larger coverage cell than that of a typical BS. This fact is known as Feller's paradox (see e.g. \cite{baccelli2000superposition}). For $x\in \Phi$, denote by $C (x, \Phi)$ the Voronoi cell (induced by $\Phi$) centered at $x$. From \cite{ferenc2007size}, we know the pdf of  the size $\| C (0, \Phi) \|$ of the coverage area of the tagged BS is given by
$$
f_{\| C (0, \Phi) \|} (x) = \frac{3.5^{4.5}}{\Gamma(4.5)} \lambda (\lambda x)^{3.5} e^{-3.5 \lambda x}, \quad x \geq 0.
$$ 
Conditioning on $\| C (0, \Phi) \|=x$, the number of UEs located in $C(0, \Phi)$ is Poisson with mean $\lambda^{(u)} x$. So the mean number of UEs located in $C(0, \Phi)$ is given by
$$
\bar{N} = \int_0^{\infty} \lambda^{(u)} x f_{\| C (0, \Phi) \|} (x) \dint x = \frac{9\lambda^{(u)}}{7\lambda}.
$$
As expected, the above term is greater than $\lambda^{(u)}/\lambda$. Then the admission probability of the typical UE in each band $i$ is given by
\begin{align}
p_i = \max \left( \frac{7\lambda B_i}{9\lambda^{(u)} b_{i}} , 1 \right).
\end{align} 
Now conditioning on the  length $\|Y_{0} \| = r$ of the typical link and the point process $\{\tilde{\Phi}_i\}$, the rate that the typical UE experiences in band $i$ equals  $p_i$ multiplied by spectral efficiency $\log (1 + \textrm{SINR}_{i}  )$. Further, conditioning on $\|Y_{0} \| = r$ and $\{\tilde{\Phi}_i\}$, these rates are independent over the $M$ bands since rate in each band $i$ only depends on the fading field in band $i$ and the fading fields are assumed to be independent over the $M$ bands. 

Summing the rates over the bands, we obtain the ergodic rate of the CA-enabled UE as 
\begin{align}
&\b E_{H}[R | \{\tilde{\Phi}_i\},   \|Y_{0} \| = r  ] \notag \\
&= \b E_{H}  \left[  \sum_{i=1}^M p_i b_i \log (1 + \textrm{SINR}_{i}  )  \Big |  \{\tilde{\Phi}_i\},  \|Y_{0} \| = r  \right] ,
\label{eq:r1:1}
\end{align}
where the expectation is over the fading fields. Denote by $I_{\tilde{\Phi}_i} = \sum_{Y \in \tilde{\Phi}_i \backslash Y_{0} } P_{i} H_{i,Y}  \parallel Y \parallel^{-\alpha_i} $, the right hand side of (\ref{eq:r1:1}) equals
\begin{align}
&\b E_{H}  \left[ \sum_{i=1}^M p_i b_i  \log \left( 1 + \frac{ P_{i} H_{i,0} r^{-\alpha_i}   }{I_{\tilde{\Phi}_i}  + w_{i}  }  \right) \right]  \notag \\
&= \sum_{i=1}^M p_i b_i \b E  \left[  \log \left( 1 + \frac{ P_{i} H_{i,0} r^{-\alpha_i}   }{I_{\tilde{\Phi}_i}  + w_{i}  }  \right) \right]  \notag \\
&= \sum_{i=1}^M  p_i b_i \int_0^{\infty}   \frac{1}{1+t} \b P \left(\frac{ P_{i} H_{i,0} r^{-\alpha_i}   }{I_{\tilde{\Phi}_i}  + w_{i}  }  \geq t \right)  \dint t, 
\label{eq:r1:2}
\end{align}
where in the last equality we use the fact: $\b E [ f(X) ] = \int_0^{\infty} f'(x) \b P( X \geq x ) \dint x $, where $f(\cdot)$ is non-negative and monotonically increasing function. Further,
\begin{align}
&\b P \left(\frac{ P_{i} H_{i,0} r^{-\alpha_i}   }{I_{\tilde{\Phi}_i}  + w_{i}  }  \geq t \right)  \notag \\
&=  \b P \left(  H_{i,0}  \geq t r^{\alpha_i} P_{i} ^{-1} ( I_{\tilde{\Phi}_i}  + w_{i}   ) \right)  
\notag \\
&=  e^{ -  t w_i r^{\alpha_i} P_{i} ^{-1}}  \b E_H \left[  e^{ - t  r^{\alpha_i} P_{i} ^{-1} I_{\tilde{\Phi}_i} } \right]   \notag \\
&=  e^{ -  t w_i r^{\alpha_i} P_{i} ^{-1}}  \b E_H \left[  e^{ - t  r^{\alpha_i}  \sum_{Y \in \Phi_{i} \backslash Y_{0} }  H_{i,Y}  \parallel Y \parallel^{-\alpha_i}   } \right]   \notag \\
&=   e^{ -  t w_i r^{\alpha_i} P_{i} ^{-1}}  \b E_H \left[ \prod_{Y \in \tilde{\Phi}_{i} \backslash Y_{0} }   e^{ - t  r^{\alpha_i}   H_{i,Y}  \parallel Y \parallel^{-\alpha_i}   } \right]   \notag \\
&=   e^{ -  t w_i r^{\alpha_i} P_{i} ^{-1}}   \prod_{Y \in \tilde{\Phi}_{i} \backslash Y_{0} } \b E_H \left[   e^{ - t  r^{\alpha_i}   H_{i,Y}  \parallel Y \parallel^{-\alpha_i}   } \right]   \notag \\
&=  e^{ -  t w_i r^{\alpha_i} P_{i} ^{-1}}  \prod_{Y \in \tilde{\Phi}_{i} \backslash Y_{0} }  \frac{1}{1 +  t r^{\alpha_i} \parallel Y \parallel^{-\alpha_i}  },
\label{eq:r1:3}
\end{align}
where in the second equality we use $H_{i,0} \sim \textrm{Exp} (1)$ and thus $\b P(H_{i,0} \geq x ) = e^{-x}$; the penultimate equality follows as the fading fields are independent; and in the last equality we use $H_{i,Y} \sim \textrm{Exp} (1)$ and thus $\b E[ e^{- s H_{i,Y} }  ] = \frac{1}{1+s}$.
Plugging (\ref{eq:r1:3}) into (\ref{eq:r1:2}) yields
\begin{align}
&\b E[R | \{\tilde{\Phi}_i\},  \{ \|Y_{0} \| = r \} ] =   \sum_{i=1}^M p_i b_i \cdot  \notag \\
&   \int_0^{\infty}  \frac{1}{1+t}  e^{ -  t w_i r^{\alpha_i} P_{i} ^{-1}}  \prod_{Y \in \tilde{\Phi}_i \backslash Y_{0} }  \frac{1}{1 +  t r^{\alpha_i} \parallel Y \parallel^{-\alpha_i}  }     \dint t. \notag 
\end{align}
Now de-conditioning with respect to $\{\tilde{\Phi}_i\}$ yields
\begin{align}
&\b E[R | \|Y_{0} \| = r ] = \b E_{\{\tilde{\Phi}_i\}} \left[\b E[R | \{\tilde{\Phi}_i\},   \|Y_{0} \| = r  ] \right] \notag \\
&=   \sum_{i=1}^M p_i b_i \int_0^{\infty}   \frac{1}{1+t}  e^{ -  t w_i r^{\alpha_i} P_{i} ^{-1}} \notag \\  & \quad \quad \quad \quad \cdot \b E_{\{\tilde{\Phi}_i\}} \left[ \prod_{Y \in \tilde{\Phi}_i \backslash Y_{0} }  \frac{1}{1 +  t r^{\alpha_i} \parallel Y \parallel^{-\alpha_i}  } \right]     \dint t \notag \\
&=   \sum_{i=1}^M p_i b_i \int_0^{\infty}   \frac{e^{ -  t w_i r^{\alpha_i} P_{i} ^{-1}}}{1+t}   e^{ - \theta_i \lambda \int_{B(0,r)} 1 - \frac{1}{1 +  t r^{\alpha_i} y^{-\alpha_i}} \dint y }     \dint t \notag \\
&=   \sum_{i=1}^M p_i b_i \int_0^{\infty}   \frac{1}{1+t}  e^{ -   w_i P_{i} ^{-1} t  r^{\alpha_i}} e^{ - \pi \theta_i \lambda \rho (t, \alpha_i, 1) r^2 }     \dint t ,
\label{eq:r1:4}
\end{align}
where in the penultimate equality the domain of integration is from $r$ to $\infty$ (because conditioning on the association with the BS at $Y_{0}$, the closest interferer has at least a distance $\| Y_{0} \|=r$ away from the typical UE), and
\begin{align}
\rho (t, \alpha, \beta)  =    \int_{1 }^\infty   \frac{t}{t + \beta x^{\frac{\alpha}{2}} }  \dint x.
\label{eq:r1:5}
\end{align}
The final step is to de-condition with respect to $\|Y_{0} \| = r$. To this end, we need the distribution function of the length $\|Y_{0} \|$  of the typical link; this is easy in single tier network \cite{baccelli2009stochastic}:
\begin{align}
\b P (\|Y_{0} \| \geq r  ) = \b P( \Phi ( B(0,r) ) = 0 ) = e^{-\lambda \pi r^2 }, r \geq 0,
\label{eq:r1:8}
\end{align} 
from which we have $f_{\|Y_{0} \| } ( r ) = 2 \pi \lambda r e^{ - \pi \lambda r^2 }, r \geq 0$. Using  $f_{\|Y_{0} \| } ( r )$ and de-conditioning with respect to $\|Y_{0} \| = r$ in (\ref{eq:r1:4}) yields
\begin{align}
\b E[R  ] &= \b E_{ \|Y_{0} \| } \left[\b E[R |  \|Y_{0} \|   ] \right] 
= \int_0^{\infty} \sum_{i=1}^M p_i b_i \int_0^{\infty}   \notag \\ 
& \cdot \frac{1}{1+t}  e^{ -   w_i P_{i} ^{-1} t  r^{\alpha_i}} e^{ - \pi \theta_i \lambda \rho (t, \alpha_i) r^2 }     \dint t \cdot 2 \pi \lambda r e^{ - \pi \lambda r^2 }  \dint r .\notag 
\end{align}
Applying Fubini's theorem, we finally obtain UE ergodic rate in Prop. \ref{pro:r3}.
\begin{pro}[\textbf{Single Tier UE Ergodic Rate}]
The single tier UE ergodic rate $\bar{R} = \b E[R  ] $ is given by 
\begin{align}
\bar{R}
= \sum_{i=1}^M 2 \pi \lambda p_i b_i  & \int_0^{\infty}   \int_0^{\infty}  \frac{1}{1+t}  e^{ -   w_i P_{i} ^{-1} t  r^{\alpha_i} } \notag \\
& \cdot e^{ - \pi \theta_i \lambda \rho (t, \alpha_i, 1) r^2 }        e^{ - \pi \lambda r^2 }  r \dint r \dint t .
\label{eq:r1:6}
\end{align}
\label{pro:r3}
\end{pro}

Further simplification is possible when the noise is ignored, i.e., $w_i \to 0$; we term this case \textit{interference-limited} networks. This special case is of interest because, compared to interference, thermal noise is often not an important issue in modern cellular networks. With $w_i \to 0$, (\ref{eq:r1:6}) reads as follows.
\begin{align}
\b E[R  ] 
&=  \sum_{i=1}^M  p_i b_i   \int_0^{\infty}      \frac{1}{(1+t) (1 + \theta_i  \rho (t, \alpha_i, 1) )}   \dint t. \notag
\end{align}
If $\lambda^{(u)} \geq \lambda$, $b_i = B_i$ and the tagged BS only serves the typical UE, the UE will have a peak data rate:
\begin{align}
\b E[R  ] 
&= \sum_{i=1}^M  q(\alpha_i)  B_i ,
\label{eq:r1:7}
\end{align}
where $q ( \alpha   ) = \int_0^{\infty}     \frac{1}{(1+t) (1 +  \rho (t, \alpha, 1) )}   \dint t$. Some remarks are in order.
\begin{itemize}
\item $q ( \alpha   ) $ is increasing with $\alpha$ as $\rho (t, \alpha, \beta) $ is decreasing with $\alpha$. This implies that  bands of larger path-loss exponents can offer higher rate per Hz, which seems a bit counter-intuitive. However, a careful thought reveals that  bands of larger path-loss exponents provide better spatial separation for the wireless links, which is particularly important in interference-limited networks.
\item Interestingly, Eq. (\ref{eq:r1:7}) does not depend on the BS density $\lambda$ and transmit powers $\{P_i\}$. This is intuitive because,  increasing  BS density $\lambda$ leads to increased signal power but also increased interference power; these two effects exactly counter-balance each other when noise is ignored. Similar reasoning holds if one increases the transmit powers. 
\item Following the previous remark, we see that increasing BS density in interference-limited networks does not lead to increased peak data rate. (But deploying more BSs does allow the network to serve more UEs at the same time.) This strongly motivates other approaches e.g. CA for increasing UE peak data rates.
\item Intuitively, it is believed that peak data rate will scale \textit{linearly} with the number of aggregated carriers \cite{Ghosh2010CA}. Eq. (\ref{eq:r1:7}) provides a finer characterization of this statement: If the aggregated carriers are sorted in ascending order based on their path-loss exponents, the peak data rate in interference-limited networks will scale \textit{super-linearly} with the number of aggregated carriers.
\end{itemize}

\section{Single Flow Carrier Aggregation}
\label{sec:single}

In this section we extend the results for single tier networks to the general $K$ tier HetNets with single flow CA. We make the following additional assumption for tractability.
\begin{assum}
The transmit power of each tier $k$ is not band-dependent, i.e., $P_{i,k} = P_k, \forall i \in \mathcal{M}$.
\label{asum:2}
\end{assum}

\subsection{Coverage, Admission Probability and Load}
\label{subsec:CovAd}

The key difference between HetNets and single tier cellular networks is that a typical UE in a HetNet can connect to any one of the $K$ tiers. 
In other words, UEs in different areas are (possibly) served by different kinds of BSs. For example, some UEs may connect to macro BSs while other UEs are served by newly deployed low power nodes including pico and femto BSs. As a first step, we characterize in Lemma \ref{lem:3} the coverage of each tier $k$, defined as the fraction of UEs served by BSs in tier $k$.
\begin{lem}[\textbf{Single Flow Coverage}]
Let $\pi = ( \pi_k  )_{k \in \mathcal{K}}$,\footnote{We use $\pi$ to denote both the single flow coverage and the constant Pi. The meaning should be clear from the context.} where $\pi_k $ is the fraction of single flow UEs served by BSs in tier $k$. Then
\begin{align}
\pi_k =2 \pi \lambda_k \int_0^\infty r h_k(r) dr,
\label{eq:902}
\end{align}
where $h_k (r) = \exp( -\pi \sum_{\ell \in \mathcal{K}} \lambda_\ell (\frac{ Z_\ell P_\ell }{Z_k P_k})^{ \frac{2}{\alpha_{\ell^\star }} } r^{\frac{ 2 \alpha_{k^\star }}{ \alpha_{\ell^\star } }}  )$ and $ k^\star = \arg \min_{i \in \mathcal{M}: x_{i,k} \neq 0} \alpha_i $ .
\label{lem:3}
\end{lem}
\begin{proof}
See Appendix \ref{proof:lem:3}.
\end{proof}
The coverage $\{\pi_k\}$ can be equivalently understood as the tier connection probabilities of the typical UE. That is, $\pi_k$ denotes the probability that the typical UE is ``covered'' by tier $k$. To gain some intuition about single flow coverage, let us sort the path-loss exponents in ascending order: $\alpha_1 \leq ... \leq \alpha_M$. Suppose all the tiers use band $1$. Then $k^\star = 1, \forall k\in \mathcal{M}$, and (\ref{eq:902}) reduces to the following:
\begin{align}
\pi_k = 1 - \frac{ \sum_{\ell \neq k} \lambda_\ell ( Z_\ell P_\ell )^{ \frac{2}{\alpha_{1 }} } }{ \lambda_k ( Z_k P_k )^{ \frac{2}{\alpha_{1 }} } + \sum_{\ell \neq k} \lambda_\ell ( Z_\ell P_\ell )^{ \frac{2}{\alpha_{1 }} }}.
\notag 
\end{align}
The above equality clearly shows that increasing biasing $Z_k$ (resp. BS density $\lambda_k$) in tier $k$ increases its coverage $\pi_k$, and $\pi_k \to 1$ as $Z_k \to \infty$ (resp. $\lambda_k \to \infty$). Further, since $\alpha_1 >2$, the coverage $\pi_k$ is more sensitive to the variation in the BS density $\lambda_k$ than to the variation in the biased transmit power $Z_k P_k$, which is also true for the throughput in wireless packet networks \cite{win2009mathematical}.

With Lemma \ref{lem:3}, we next derive the load of each tier, which uniformly and independently blocks UEs in its coverage if it is fully-loaded. 
\begin{lem}[\textbf{Single Flow Load}]
With single flow CA, the load $\theta_{i,k}$ of tier $k$ in band $i$ is given by
\[ \theta_{i,k} = \left\{ \begin{array}{ll}
         \frac{   2 \pi b_{i,k} \lambda^{(u)} G_k  }{B_i }  & \mbox{if $x_{i,k}\neq 0 \textrm{ and } 0<    2 \pi \lambda^{(u)} G_k  \leq \frac{B_i}{b_{i,k}}$};\\
         1 & \mbox{if $x_{i,k}\neq 0 \textrm{ and } 2 \pi \lambda^{(u)} G_k  > \frac{B_i}{b_{i,k}}$};\\
         0 & \mbox{if $x_{i,k}=0$}.\end{array} \right. \]
where $G_k = \int_0^\infty r h_k(r) dr$.
\label{lem:4}
\end{lem}
\begin{proof}
See Appendix \ref{proof:lem:4}.
\end{proof}

Finally, let us consider the admission probability. In $K$-tier HetNets, the mean number of UEs served by the tagged BS depends on which tier it belongs to and is hard to compute exactly. In this paper, we adopt the following approximation proposed in \cite{singh2012offloading} for the mean number of UEs served by the tagged BS in tier $k$:
\begin{align}
\bar{N}_k \cong 1 + \frac{1.28\pi_k \lambda^{(u)}}{\lambda_k}.
\label{eq:r1:18}
\end{align}
With this approximation the admission probability $p_{i,k}$ of tier $k$ in band $i$ is given by
\[ p_{i,k} = \left\{ \begin{array}{ll}
         1  & \mbox{if $x_{i,k}\neq 0 \textrm{ and } 0<     \bar{N}_k  \leq \frac{B_i}{b_{i,k}}$};\\
         \frac{B_i}{b_{i,k}\bar{N}_k }  & \mbox{if $x_{i,k}\neq 0 \textrm{ and } \bar{N}_k > \frac{B_i}{b_{i,k}}$};\\
         0 & \mbox{if $x_{i,k}=0$}.\end{array} \right. \]

\subsection{UE Ergodic Rate}

In this subsection, we extend the proof technique used in Section \ref{sec:multiband} to derive single flow UE ergodic rate in general $K$-tier networks. To begin with, let $J$ be the random tier that the typical UE connects to. Conditioning on $J=k$, $\|Y_{k_0} \| = r$, and  $\{\tilde{\Phi}_{i,\ell}\}$, we obtain the ergodic rate of the CA-enabled UE as 
\begin{align}
&\b E_{H}[R | \{\tilde{\Phi}_{i,\ell}\},   \|Y_{k_0} \|  = r, J=k  ]  \notag \\
&= \b E_{H}  \left[  \sum_{i=1}^M p_{i,k} b_{i,k} \log (1 + \textrm{SINR}_{i,k}  )  \Big |  \{\tilde{\Phi}_{i,\ell}\},  \|Y_{k_0} \| = r, J=k \right] \notag \\
&= \b E_{H}  \left[  \sum_{i=1}^M p_{i,k} b_{i,k} \log \left(1 + \frac{ P_{k} H_{i,k_0} r^{-\alpha_i}   }{ \sum_{\ell \in \mathcal{K}}  I_{\tilde{\Phi}_{i,\ell}} + w_{i,k}  }  \right)  \right] \notag \\
&=   \sum_{i=1}^M p_{i,k} b_{i,k} \int_0^{\infty} \frac{1}{1+t}  \b P \left( \frac{ P_{k} H_{i,k_0} r^{-\alpha_i}   }{\sum_{\ell \in \mathcal{K}}  I_{\tilde{\Phi}_{i,\ell}} + w_{i,k} } \geq t   \right) \dint t, \notag
\end{align}
where 
$
I_{\tilde{\Phi}_{i,\ell}} = \sum_{Y \in \tilde{\Phi}_{i,\ell} \backslash Y_{k_0} } P_{\ell} H_{i,Y}  \parallel Y \parallel^{-\alpha_i}
$ and
\begin{align}
&\b P \left( \frac{ P_{k} H_{i,k_0} r^{-\alpha_i}   }{\sum_{\ell \in \mathcal{K}}  I_{\tilde{\Phi}_{i,\ell}} + w_{i,k} } \geq t   \right)  \notag \\
&=  e^{ -  t w_{i,k} r^{\alpha_i} P_{k} ^{-1}} \prod_{\ell}  \prod_{Y \in \tilde{\Phi}_{i,\ell} \backslash Y_{0} }  \frac{1}{1 +  t r^{\alpha_i} P_k^{-1} P_{\ell} \parallel Y \parallel^{-\alpha_i}  }.
\notag 
\end{align}
Now de-conditioning with respect to $\{\tilde{\Phi}_{i,\ell}\}$ yields
\begin{align}
&\b E[R | \|Y_{k_0} \| = r, J=k ] =   \sum_{i=1}^M p_{i,k} b_{i,k} \int_0^{\infty} \frac{1}{1+t} \notag \\
 &  \cdot e^{ -  t w_{i,k} r^{\alpha_i} P_{k} ^{-1}}  e^{ - \sum_{\ell} \theta_{i,\ell} \lambda_\ell \int_{B(0,\xi_{k,\ell})} 1 - \frac{1}{1 +  t r^{\alpha_i} P_k^{-1} P_{\ell} \|y\|^{-\alpha_i}} \dint y }     \dint t, \notag
\end{align}
where $\xi_{k,\ell}$ is determined by
$
Z_k P_{k} \cdot r^{-\alpha_{k^\star}}  = Z_\ell P_{\ell} \cdot \xi_{k,\ell}^{-\alpha_{\ell^\star}}
$, from which we have $\xi_{k,\ell} = (\frac{Z_\ell P_{\ell}}{Z_k P_{k}})^{\frac{1}{\alpha_{\ell^\star}}} r^{\frac{\alpha_{k^\star}}{\alpha_{\ell^\star}}}$.
It follows that
\begin{align}
&\sum_{\ell} \theta_{i,\ell} \lambda_\ell \int_{B(0,\xi_{k,\ell})} 1 - \frac{1}{1 +  t r^{\alpha_i} P_k^{-1} P_{\ell} \|y\|^{-\alpha_i}} \dint y  \notag \\
&= \sum_{\ell} \pi \theta_{i,\ell}  \lambda_\ell \xi_{k,\ell}^2 \rho \left(t, \alpha_i, \frac{P_k}{P_l} \left(\frac{\xi_{k,\ell}}{r} \right)^{\alpha_i} \right). \notag 
\end{align}

The next step is to de-condition with respect to $\|Y_{k_0} \| = r$. To this end, we need to derive the distribution of the distance $\|Y_{k_0} \|$ \textit{conditioning on $J=k$}, which is given in Lemma \ref{lem:5}.
\begin{lem}
Conditioning on $J=k$, the pdf of the distance $\|Y_{k_0} \|$ is given by
\begin{align}
f_{\|Y_{k_0} \| \mid J } (r) = \frac{1}{\pi_k} 2 \pi \lambda_k  r \cdot e^{ -\pi \sum_{\ell \in \mathcal{K}} \lambda_\ell (\frac{ Z_\ell P_\ell }{Z_k P_k})^{ \frac{2}{\alpha_{\ell^\star}} } r^{\frac{2 \alpha_{k^\star}}{ \alpha_{\ell^\star} }}  }  , r \geq 0.
\label{eq:r1:12}
\end{align}
\label{lem:5}
\end{lem}
\begin{proof}
See Appendix \ref{proof:lem:5}.
\end{proof}
To gain some intuition about the conditional distribution of $\|Y_{k_0} \|$, suppose $\alpha_1 \leq ... \leq \alpha_M$ and that all the tiers use band $1$. Then (\ref{eq:r1:12}) reduces to the following Rayleigh distribution:
\begin{align}
&f_{\|Y_{k_0} \| \mid J } (r) = \notag \\
&2 \pi r \sum_{\ell \in \mathcal{K}} \lambda_\ell (\frac{ Z_\ell P_\ell }{Z_k P_k})^{ \frac{2}{\alpha_{1}} }   \exp( - \pi r^2 \sum_{\ell \in \mathcal{K}} \lambda_\ell (\frac{ Z_\ell P_\ell }{Z_k P_k})^{ \frac{2}{\alpha_{1}} } ), \quad r\geq 0.
\notag 
\end{align}
In particular, its first moment is given by $\b E[ \|Y_{k_0} \| \big | J ] =\sqrt{\pi_k} \cdot \frac{1}{2} \sqrt{\frac{1}{\lambda_k}}$. Note that $\frac{1}{2} \sqrt{\frac{1}{\lambda_k}}$ is the mean distance to the closest BS in tier $k$ from the typical UE \textit{without} conditioning. Interestingly, the mean length of the typical radio link in tier $k$ conditional on $J=k$ equals the unconditional mean $\frac{1}{2} \sqrt{\frac{1}{\lambda_k}}$ multiplied by a factor $\sqrt{\pi_k} \leq 1$. Thus, conditioning reduces the mean distance to the closest BS, agreeing with intuition.

Using Lemma \ref{lem:5}, we now can de-condition on $\|Y_{k_0} \| = r$ and obtain that $\b E[R | J=k ]$ equals
\begin{align}
\sum_{i=1}^M p_{i,k} b_{i,k} & \frac{2\pi \lambda_k }{\pi_k} \int_0^{\infty} \int_0^{\infty}  \frac{1}{1+t}  e^{ -  t w_{i,k} r^{\alpha_i} P_{k} ^{-1}}  \notag \\
& \cdot e^{ - \pi \sum_{\ell}   \lambda_\ell \xi_{k,\ell}^2 \left(\theta_{i,\ell} \rho \left(t, \alpha_i, \frac{P_k}{P_l} \left(\frac{\xi_{k,\ell}}{r} \right)^{\alpha_i} \right) +  1 \right)  } r  \dint r  \dint t. \notag
\end{align}
As a final step, we de-condition on $J = k$ and obtain the following Prop. \ref{pro:r2}.
\begin{pro}[\textbf{Single Flow UE Ergodic Rate}]
The single flow UE ergodic rate is given by 
\begin{align}
\bar{R} 
&=  \sum_{i \in \mathcal{M}}  \sum_{k \in \mathcal{K}} 2 \pi  \lambda_k b_{i,k} p_{i,k}  \int_0^\infty  \int_0^\infty  \frac{g^{(s)}_{i,k} (t, r) h_k (r) r}{1+t}    \dint r \dint t,  
\label{eq:sm}
\end{align}
where  $g_{i,k}^{(s)} ( t , r )$ equals
\begin{align}
e^{ -  t w_{i,k} r^{\alpha_i} P_{k} ^{-1}}  e^{ - \pi \sum_{\ell}   \lambda_\ell   \theta_{i,\ell} (\frac{Z_\ell P_{\ell}}{Z_k P_{k}})^{\frac{2}{\alpha_{\ell^\star}}} \rho \left(t, \alpha_i, \frac{P_k}{P_l} \left(\frac{\xi_{k,\ell}}{r} \right)^{\alpha_i} \right) r^{2\frac{\alpha_{k^\star}}{\alpha_{\ell^\star}}} } .
\label{eq:pro22} 
\end{align}
with $\xi_{k,\ell} = (\frac{Z_\ell P_{\ell}}{Z_k P_{k}})^{\frac{1}{\alpha_{\ell^\star}}} r^{\frac{\alpha_{k^\star}}{\alpha_{\ell^\star}}}$.
\label{pro:r2}
\end{pro}

\section{Discussions on Network Deployment}

In this section, we use the derived analytical results to study the HetNet performance through examining two typical band deployment scenarios: orthogonal and cochannel deployment.

\subsection{Orthogonal Deployment}
\label{subsec:ortho}

In this part we assume orthogonal deployment: different tiers use different bands. Without loss of generality, suppose $M=K$ and tier $k$ is matched with  band $k$.   Then Prop. \ref{pro:r2} reduces to the following Corollary \ref{cor:4}.
\begin{cor}
Suppose $M=K$ and orthogonal deployment, i.e., $x_{i,k} = 1$ if $i=k$ and $0$ otherwise. Then  the single flow UE ergodic rate is given by $\bar{R} =   \sum_{k \in \mathcal{K}} \bar{R}_k $, where $\bar{R}_k$ is the rate served by tier $k$ and equals
\begin{align}
\bar{R}_k 
= \pi_k &  b_{k,k} p_{k,k}  \int_0^\infty  \int_0^\infty  \frac{1}{1+t}   e^{ -  t w_{k,k} r^{\alpha_k} P_{k} ^{-1}}  \notag \\
& \cdot e^{ - \pi \lambda_k  \theta_{k,k} \rho \left(t, \alpha_k, 1\right) r^{2} }   f_{\|Y_{k_0} \| \mid  J } (r)  \dint r \dint t,
\label{eq:r1:13}
\end{align}
where $f_{\|Y_{k_0} \| \mid  J } (r)$ is given in Lemma \ref{lem:5} with $\ell^{\star}$ replaced by $\ell$.
Further, if noise is ignored, 
\begin{align}
\bar{R} = \sum_{k \in \mathcal{K}}  \pi_k  b_{k,k} p_{k,k}  \int_0^\infty \frac{1}{1+t} \frac{1}{ 1 + \pi_k \theta_{k,k}  \rho \left(t, \alpha_k, 1\right)  } \dint t. \notag 
\end{align}
\label{cor:4}
\end{cor}

With Corollary \ref{cor:4}, we first discuss the impact of biasing  by examining its impact on each term in (\ref{eq:r1:13}). For ease of exposition, we further assume $\alpha_k = \alpha, \forall k$.

\subsubsection{Impact on coverage} As pointed out in Section \ref{subsec:CovAd}, $\pi_k = 1 - \frac{ \sum_{\ell \neq k} \lambda_\ell ( Z_\ell P_\ell )^{ \frac{2}{\alpha} } }{ \lambda_k ( Z_k P_k )^{ \frac{2}{\alpha} } + \sum_{\ell \neq k} \lambda_\ell ( Z_\ell P_\ell )^{ \frac{2}{\alpha} }}$. Thus, increasing biasing factor $Z_k$  increases the coverage $\pi_k$ of tier $k$. On the contrary, $\pi_{\ell}, \ell \neq k$, decreases with $Z_k$ and thus the coverage areas of the other tiers shrink. 

\subsubsection{Impact on admission probabilities} Recall that $p_{k,k}=1$ if tier $k$ has limited coverage such that $\bar{N}_k \leq \frac{B_k}{b_{k,k}}$; otherwise, $p_{k,k}=\frac{B_k}{b_{k,k} \bar{N}_k }$. So if tier $k$ is under-loaded, $p_{k,k}$ remains $1$ with increasing $Z_k$ until $\bar{N}_k = \frac{B_k}{b_{k,k}}$, after which $p_{k,k}$ decreases as $Z_k$ further increases. On the contrary, $p_{\ell, \ell}, \ell \neq k$, is monotonically non-decreasing with increasing $Z_k$ and thus UEs connecting to tier $\ell$ can be scheduled more often, at least unchanged.

\subsubsection{Impact on interference powers} The interfering transmitter density equals $\lambda_k  \theta_{k,k} $. From Lemma \ref{lem:4}, $\theta_{k,k}=\frac{b_{k,k} \lambda^{(u)} \pi_k }{B_k\lambda_k}$ if tier $k$ has limited coverage such that $\pi_k \leq \frac{B_k\lambda_k}{b_{k,k} \lambda^{(u)}}$; otherwise, $\theta_{k,k}=1$. So if tier $k$ is under-loaded, $\theta_{k,k}$ (resp. interference in tier $k$) increases with increasing $Z_k$ until $\pi_k = \frac{B_k\lambda_k}{b_{k,k} \lambda^{(u)}}$, after which $\theta_{k,k}\equiv 1$ (resp. interference in tier $k$ remains constant)  as $Z_k$ further increases. Note the impact of interference on $R_k$ is shown by the term $e^{ - \pi \lambda_k  \theta_{k,k} \rho \left(t, \alpha_k, 1\right) r^{2} }$ in (\ref{eq:r1:13}). On the contrary, $\theta_{\ell, \ell}, \ell \neq k$, is monotonically non-increasing with increasing $Z_k$ and thus interference in tier $\ell$ gets decreased, at least unchanged.

\subsubsection{Impact on the lengths of radio links}  From $f_{\|Y_{k_0} \| \mid  J }$ given in Lemma \ref{lem:5}, we know that 
\begin{align}
\b P( \|Y_{k_0} \| \geq r | J = k ) =  \exp( - \pi r^2 \sum_{\ell \in \mathcal{K}} \lambda_\ell (\frac{ Z_\ell P_\ell }{Z_k P_k})^{ \frac{2}{\alpha_{1}} } ),
\label{eq:r1:16}
\end{align}
which is increasing as $Z_k$ increases. So increasing biasing factor $Z_k$ makes tier $k$ serve more UEs of longer radio links. On the contrary, $\b P( \|Y_{\ell_0} \| \geq r | J = k ) $ is decreasing as $Z_k$ increases; and thus tier $\ell, \ell \neq k$, serves more UEs of shorter radio links.

To sum up, the overall impact of biasing on UE ergodic rate depends on each tier's coverage, admission probability, interference, and lengths of radio links. For a given arbitrary network, it is not \textit{a priori} clear whether biasing improves or hurts its performance. Nevertheless, it is generally believed that increasing the biasing of small cells  makes them accomplish more as loads become more balanced over the tiers; thus, biasing UEs towards small cells can benefit the network as a whole.

\subsection{Cochannel Deployment}
\label{subsec:cochannel}

In this part we focus on single band case and drop the subscript $i$ for ease of notation. Then Prop. \ref{pro:r2} reduces to the following Corollary \ref{cor:5}.
\begin{cor}
Suppose that a single band with path-loss exponent $\alpha$ is used by all the $K$ tiers. Then the UE ergodic rate is given by $\bar{R} =   \sum_{k \in \mathcal{K}} \bar{R}_k $ where 
\begin{align}
&\bar{R}_k 
= \pi_k   b_{k} p_{k}  \int_0^\infty  \int_0^\infty  \frac{1}{1+t}   e^{ -  t w_{k} r^{\alpha} P_{k} ^{-1}}  \notag \\
& \cdot e^{   - r^2 \cdot \pi \sum_{\ell \in \mathcal{K}} \theta_{\ell} \lambda_\ell (\frac{ Z_\ell P_\ell }{Z_k P_k})^{ \frac{2}{\alpha} } \rho(t, \alpha, \frac{Z_\ell}{Z_k})  }   f_{\|Y_{k_0} \| \mid  J } (r)  \dint r \dint t.
\label{eq:r1:17}
\end{align}
where $f_{\|Y_{k_0} \| \mid  J } (r)$ is given in Lemma \ref{lem:5} with $\alpha_{\ell^{\star}}$ replaced by $\alpha$.
Further, if noise is ignored, 
\begin{align}
\bar{R} = \sum_{k \in \mathcal{K}}  \pi_k  b_{k} p_{k}  \int_0^\infty \frac{1}{1+t} \frac{1}{ 1 +  \sum_{\ell} \pi_\ell \theta_\ell  \rho(t, \alpha, \frac{Z_\ell}{Z_k})   } \dint t. \notag 
\end{align}
\label{cor:5}
\end{cor}

With Corollary \ref{cor:5}, we now discuss the impact of biasing  by examining its impact on each term in (\ref{eq:r1:17}). It is not difficult to see the impact of biasing on coverage, admission probabilities and the lengths of radio links is the same as the orthogonal deployment case; the difference lies in the interference, on which we shall focus below.

In cochannel case, the interference power experienced by the UEs served by tier $k$ is proportional to the exponent of the third term of the integrand in (\ref{eq:r1:17}), which can be re-written as
\begin{align}
r^2 \pi \left(  \theta_k \lambda_k \rho(t, \alpha, 1)  + \sum_{\ell \neq k} \theta_{\ell} \lambda_\ell (\frac{P_\ell }{P_k})^{ \frac{2}{\alpha} } \int_{\frac{Z_\ell}{Z_k}}^{\infty} \frac{t}{t+x^{\frac{\alpha}{2}}} \dint x  \right).
\notag 
\end{align}
Here the first term in the above parentheses indicates the intra-tier interference level and is non-decreasing with increasing $Z_k$. If tier $k$ is already fully-loaded, i.e., $\theta_k = 1$, increasing biasing does not change the intra-tier interference; otherwise, increasing biasing makes BSs of tier $k$ more active and thus increases intra-tier interference.

As for the inter-tier interference e.g. from tier $\ell$, the impact of increasing $Z_k$ is more subtle: as $Z_k$ increases, $\theta_{\ell}$  decreases but $\int_{{Z_\ell}/{Z_k}}^{\infty} \frac{t}{t+x^{\frac{\alpha}{2}}} \dint x$ increases. To explain this subtle phenomenon,
we partition the UEs connecting to tier $k$ into two groups: Group I consists of the original UEs connecting to tier $k$ before increasing  $Z_k$ and Group II consists of the new UEs connecting to tier $k$ which are biased from other tiers due to the increased  $Z_k$. Increasing $Z_k$ makes other tiers  less active and thus UEs of Group I will experience less interference; in contrast, increasing $Z_k$ brings UEs of Group II closer to the inter-tier interferers, though interferers are less active. 
To sum up, increasing $Z_k$ can either increase or decrease the inter-tier interference from tier $\ell$; the answer depends on the trade-off between the two factors mentioned above. As a special case, if tier $\ell$ remains fully loaded when $Z_k$ increases moderately, then  the inter-tier interference from tier $\ell$ increases.

Next we compare Corollary \ref{cor:5} (i.e., $1$-band-$K$-tier deployment) to Prop. \ref{pro:r3} (i.e., $K$-band-$1$-tier deployment). This comparison is interesting as it will demonstrate the advantages/disadvantages of two popular approaches for increasing rate in cellular networks: spatial reuse with small cells versus adding more bandwidth. To this end, the following result that follows from Prop. \ref{pro:r3} and  Corollary \ref{cor:5} is instrumental.
\begin{pro}
Suppose the following assumptions are satisfied:
\begin{enumerate}
\item The network is interference-limited, i.e., noise is ignored;
\item The UE density $\lambda^{(u)}$ is large enough;
\item No biasing factors are applied, i.e., $Z_k = 1, \forall k$; \footnote{This is a natural assumption in a fully-loaded network for maximizing sum rate.}
\item All the bands have the same path-loss exponent $\alpha$ and bandwidth $B$;
\end{enumerate}
Then the UE ergodic rates of $1$-band-$K$-tier and $K$-band-$1$-tier  deployments  are respectively given by
\begin{align}
\bar{R}_{1-K} &=\sum_{k=1}^K \lambda_k \cdot \frac{0.78 B}{\lambda^{(u)}}\int_0^{\infty} \frac{1}{1+t} \frac{1}{1+ \rho(t,\alpha,1)} \dint t \notag \\
\bar{R}_{K-1} &=K \lambda_1 \cdot \frac{0.78 B}{\lambda^{(u)}}\int_0^{\infty} \frac{1}{1+t} \frac{1}{1+ \rho(t,\alpha,1)} \dint t . 
\end{align}
\label{pro:1}
\end{pro}
Interestingly, $K$-tier-$1$-band deployment can have higher rate than that of  $1$-tier-$K$-band deployment as long as $\frac{\sum_{k=1}^K \lambda_k}{K \lambda_1} \geq 1$ which can be satisfied by making $\lambda_k \geq \lambda_1, \forall k$. Further,
the gain $\frac{\sum_{k=1}^K \lambda_k}{K \lambda_1}$ increases proportionally with the BS densities $\lambda_k, k \neq 1$; in contrast, the gain decreases if $\lambda_1$ increases, agreeing with intuition: small cells become less useful if the macro BSs are deployed more densely. To sum up, Prop. \ref{pro:1} gives a theoretical support to the current vast interest in adding small cells to cellular networks. In particular, it demonstrates the potential of small cells in solving the current ``spectrum crunch''. 

With similar assumptions as in Prop. \ref{pro:1}, we can also derive the UE ergodic rate of orthogonal $K$-band-$K$-tier deployment as
\begin{align}
\bar{R}^{(\perp)}_{K-K} &=\sum_{k=1}^K \lambda_k \cdot \frac{0.78 B}{\lambda^{(u)}}\int_0^{\infty} \frac{1}{1+t} \frac{1}{1+ \pi_k \rho(t,\alpha,1)} \dint t, \notag
\end{align}
where the superscript $\perp$ denotes orthogonal deployment. Similarly, the UE ergodic rate $\bar{R}^{(co)}_{K-K}$ of cochannel $K$-band-$K$-tier  deployment is given by
\begin{align}
\bar{R}^{(co)}_{K-K} &=\sum_{k=1}^K \lambda_k \cdot \frac{0.78 B}{\lambda^{(u)}}\int_0^{\infty} \frac{1}{1+t} \frac{K}{1+ \rho(t,\alpha,1)} \dint t. \notag 
\end{align}
Typically, $\bar{R}^{(co)}_{K-K} \geq \bar{R}^{(\perp)}_{K-K}$. However, the performance gap can be reduced by exploiting the additional design dimension of orthogonal $K$-band-$K$-tier deployment: One can appropriately tune $\{\pi_k\}$ (e.g. by optimizing biasing factors) to get larger rate.

\section{Simulation and Numerical Results}
\label{sec:sim}

In this section, we provide numerical results to demonstrate the analytical results. The main set-up is a typical HetNet consisting of 2 tiers (macro and small BSs) and 2 bands ($800$MHz and $2.5$GHz band). The specific parameters used are summarized in Table \ref{tab:sys:para} unless otherwise specified. We further introduce a shorthand notation $[x_{1,1}, x_{2,1}; x_{1,2}, x_{2,2}]$ to denote the band deployment configuration for ease of description. For example, configuration $[1, 1; 1, 0]$ denotes that tier 1 uses both band 1 and band 2 while tier 2 only uses  band 1.

\begin{table}
\caption{System Parameters}
\centering
\begin{tabular}{|l||r|} \hline
Density of macro BSs $\lambda_1$ & $(\pi 500^2)^{-1}$ m$^{-2}$  \\ \hline 
Density of small BSs $\lambda_2$ & $2 \times (\pi 500^2)^{-1}$ m$^{-2}$ \\ \hline
Density of UEs $\lambda^{(u)}$   & $\lambda^{(u)} = 12  \lambda_2$ \\ \hline 
Max Tx power of macro BSs & $40$ W \\ \hline
Max Tx power of small BSs & $1$ W \\ \hline
Noise PSD & $-174$ dBm \\ \hline 
Noise figure & $6$ dB \\ \hline \hline
$800$MHz band's wavelength $\mu_1$ & 0.375 m \\ \hline
$800$MHz band's path loss exp. $\alpha_1$ & 3 \\ \hline
$800$MHz band's  bandwidth $B_1$ & $9$ MHz \\ \hline
$800$MHz band's  biasing $Z_1$   & $0$ dB \\ \hline
$2.5$GHz band's wavelength $\mu_2$ & 0.12 m \\ \hline
$2.5$GHz band's path loss exp. $\alpha_2$ & 4 \\ \hline
$2.5$GHz band's  bandwidth $B_2$ & $9$ MHz \\ \hline 
$2.5$GHz band's  biasing $Z_2$   & $0$ dB \\ \hline 
UE min BW req. $b_{i,k} \equiv b$ & $1.8$ MHz \\ \hline \hline 
Band deployment &  $[1, 0; 0, 1]$ \\ \hline 
\end{tabular}
\label{tab:sys:para}
\end{table}

\subsection{Monte Carlo Simulations}

We first validate the derived rate expressions via Monte Carlo simulations. The set-up is a square area $A$ of size $|A| = 20\times 20$ km$^2$, where $\lambda_1 = 0.5\times (\pi 500^2)^{-1}$, $\lambda_2 = 2 \lambda_1$, and $\lambda^{(u)}$ ranges from $5\lambda_1$ to $49\lambda_1$. With this set-up, the mean numbers of macro BSs and small BSs are $255$ and $510$, respectively, and   the mean number of UEs ranges from $1275$ to $12495$. Such a large network is simulated to make boundary effect negligible. Then for each $\lambda^{(u)}$, the simulation steps are as follows.
\begin{enumerate}
\item Generate a random number $N_k$ for each tier $k$ such that $N_k \sim \textrm{Poisson} ( \lambda_k  |A|)$.
\item Generate $N_k$ points that are uniformly distributed in $A$; these $N_k$ points represent the BSs in tier $k$.
\item With a similar approach as in the previous two steps, generate $N^{(u)}$ points representing the UEs.
\item Generate independently the fading gain from each BS to each UE.
\item Associate each UE to some BS using the single flow association policy introduced in Section \ref{subsec:user}.
\item For each BS, UE scheduling is done as follows: If there are more than $B_i/b = 5$ associated UEs, the BS randomly picks up $5$ out of them to serve; otherwise, there are more subchannels than the number of UEs (say $n$) associated and correspondingly the BS randomly picks up $n$  out of the $5$ subchannels to serve these $n$ UEs.
\item Compute the rate of each UE: The rate equals $b \log (1 + \sinr )$ if the UE is scheduled; otherwise, the rate is $0$.
\item Compute the average rate $\bar{R}(t)$ (where $t$ is the iteration index) by averaging over all the UE rates.
\item Repeat the above steps for $100$ times and then compute the ergodic rate as $\bar{R} = \frac{1}{100} \sum_{t=1}^{100} \bar{R}(t)$.
\end{enumerate}

The comparison results are shown in Fig. \ref{fig:4}, where we simulate each band separately. This is justified by the analysis in Section \ref{sec:multiband} which shows that spatial correlation naturally decouples over the bands when single flow CA is supported. Fig. \ref{fig:4} shows that the analytical results match the empirical results fairly well in the single tier network. However, small gaps exist between analysis and simulation in the two-tier network. These gaps arise mainly because of the approximation  (\ref{eq:r1:18}) used in the rate expression of the multi-tier networks.
Nevertheless, the analytical results provide a pretty good estimate for the simulation results, and is especially attractive for studying large networks where simulation can be quite time-consuming.

\begin{figure}
\centering
\includegraphics[width=8.5cm]{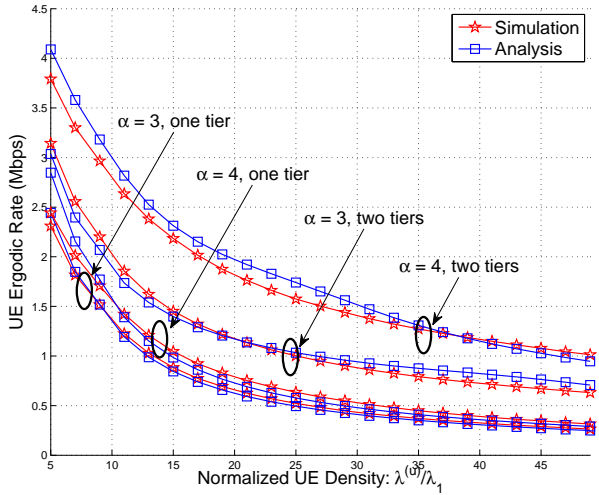}
\caption{Analytical results versus Monte Carlo simulation results}
\label{fig:4}
\end{figure}

\subsection{Biasing Effect}

In this subsection, we explore the biasing effect on network performance. First, we show the network sum rate, which equals UE ergodic rate multiplied by UE density, as a function of normalized density of small cells in Fig. \ref{fig:6}. It is shown that \textit{network capacity increases almost linearly as the density of small cells increases}, which shows the promise of small cells. Further, appropriate biasing values (e.g. 10dB in Fig. \ref{fig:6}) can  help small cells accomplish more by increasing the sum ate. Interestingly, biasing does not help when the densities of small cells (proportional to UE densities) are either too small or too large: In the former case the load of macro BSs is not heavy and thus offloading is not needed; in the latter case the load of small BSs is already heavy enough and thus offloading can hardly increase the sum rate.

\begin{figure}
\centering
\includegraphics[width=8cm]{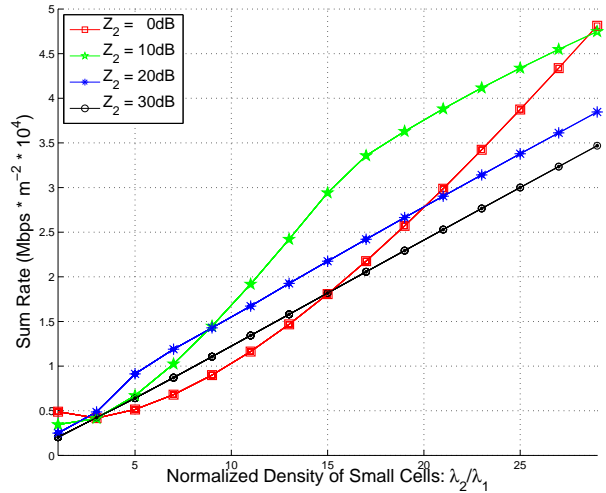}
\caption{Impact of the density of small cells and biasing on network sum rate under orthogonal deployment.}
\label{fig:6}
\end{figure}

We next provide a finer view about the effect of biasing in Fig. \ref{fig:71} and \ref{fig:72}. For brevity, we fix the UE density as $\lambda^{(u)}/\lambda_1=24$ and assume all the channels are in the 800MHz frequency band. As shown in Fig. \ref{fig:71}, the impact of biasing on coverage and load is consistent with our previous high level discussions in Section \ref{subsec:ortho}. The impact of biasing on spectral efficiency, a function of interference powers and lengths of radio links, is shown in Fig. \ref{fig:72}. As implied by the discussion in Section \ref{subsec:ortho}, with increasing biasing of small cells, the spectral efficiency of small cells decreases while the converse is true for the spectral efficiency of macro cells. However, the situation becomes more complicated when it comes to the cochannel deployment case. For example, with increasing biasing of small cells, the spectral efficiency of small cells first decreases and then increases. This subtle effect of biasing can be understood along the reasoning presented in Section \ref{subsec:cochannel}.

\begin{figure}
\centering
\includegraphics[width=8cm]{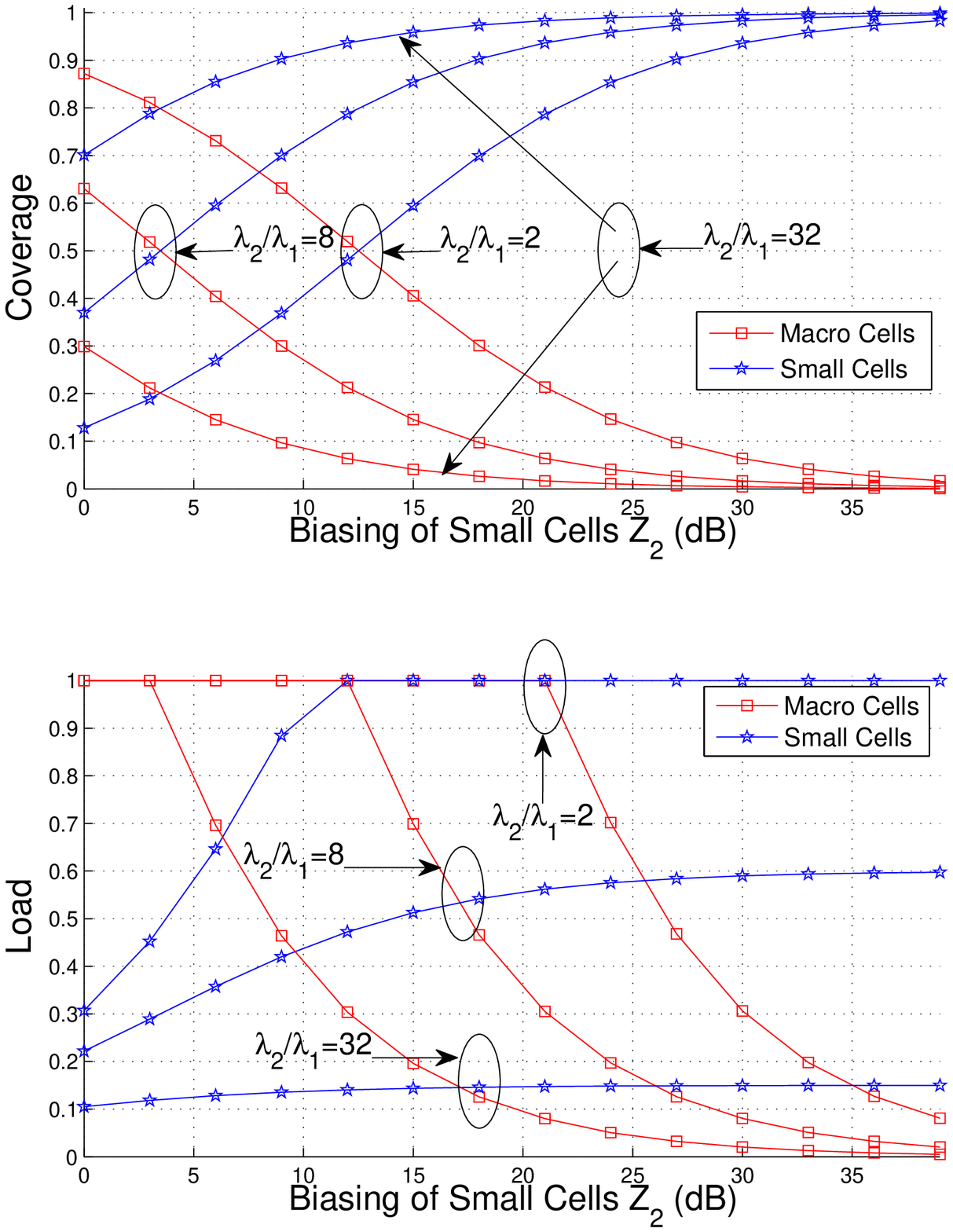}
\caption{Impact of biasing on coverage and load under orthogonal deployment.}
\label{fig:71}
\end{figure}

\begin{figure}
\centering
\includegraphics[width=8cm]{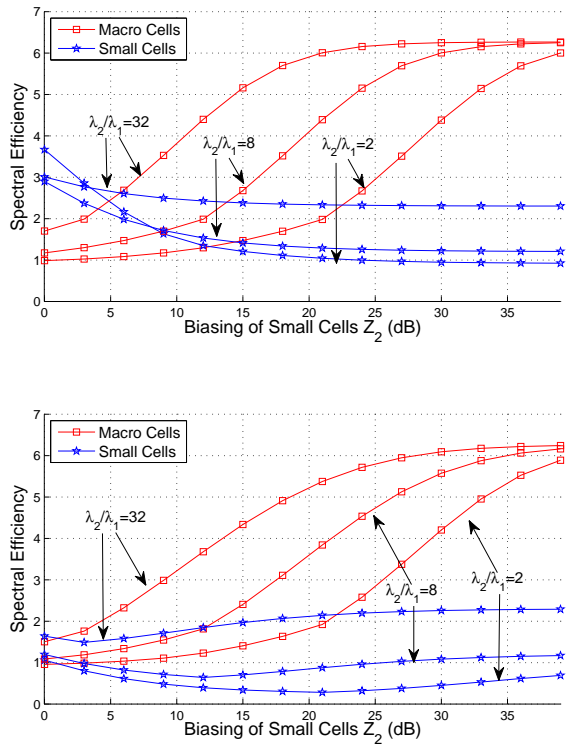}
\caption{Impact of biasing on spectral efficiency: The top subfigure is for orthogonal deployment while the bottom subfigure is for cochannel deployment.}
\label{fig:72}
\end{figure}

\subsection{Band Deployment}

In this subsection, we study how different band deployment configurations affect UE ergodic rate. Since in reality the 800MHz band has to be deployed in macro cells to ensure coverage, we only focus on those band deployments with $x_{1,1}=1$.

Fig. \ref{fig:66} and \ref{fig:1101} show the UE ergodic rates under band deployments $[1, 0; 0, 1]$ and $[1, 1; 0, 1]$, respectively. Under these two deployments, small cells do not use the 800MHz band and thus have much smaller coverage than macro cells. Hence, biasing can increase the UE ergodic rate by expanding the coverage areas of small cells; this is true for moderate densities of small cells, as shown in Fig. \ref{fig:66} and \ref{fig:1101}. For example, in the case of $\lambda_2/\lambda_1=8$ shown in Fig. \ref{fig:66}, an almost $2\times$ rate gain is achieved with $17$dB biasing.  Further, the optimal biasing value is sensitive to the  densities of small cells; in general, it decreases as $\lambda_2/\lambda_1$ increases. Note that Fig. \ref{fig:66} and \ref{fig:1101} agree with Fig.  \ref{fig:6} that biasing does not help when the densities of small cells and UEs are either too small or too large.

\begin{figure}
\centering
\includegraphics[width=8cm]{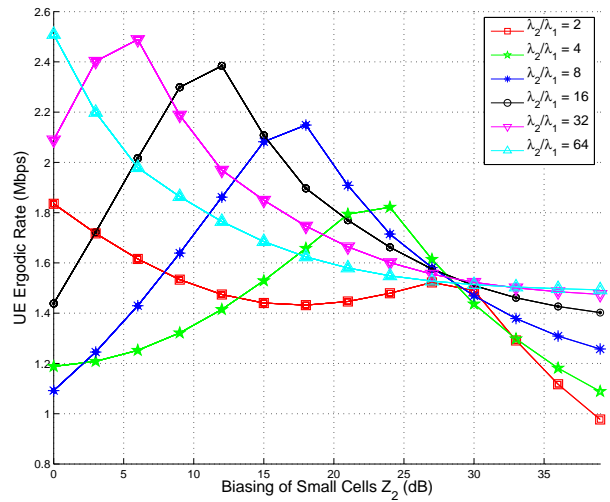}
\caption{UE ergodic rate vs. biasing under band deployment $[1, 0; 0, 1]$}
\label{fig:66}
\end{figure}

\begin{figure}
\centering
\includegraphics[width=8cm]{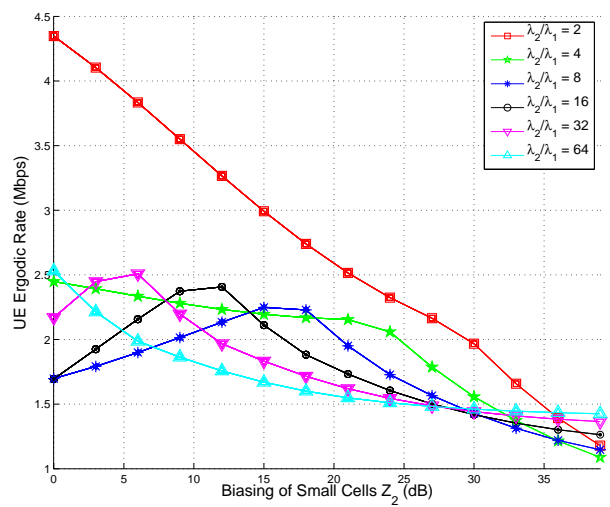}
\caption{UE ergodic rate vs. biasing under band deployment $[1, 1; 0, 1]$}
\label{fig:1101}
\end{figure}

In contrast, Fig. \ref{fig:1110} shows the UE ergodic rates under band deployment $[1, 1; 1, 0]$, where small cells use the 800MHz band instead of the 2.5GHz band. In this case, small cells have much better coverage than their counterparts under $[1, 0; 0, 1]$ and $[1, 1; 0, 1]$. Thus, biasing is not that helpful; indeed, zero biasing value is optimal as shown in Fig. \ref{fig:1110}. However, the UE ergodic rates under $[1, 1; 1, 0]$ are lower than those under $[1, 0; 0, 1]$ or $[1, 1; 0, 1]$. That is, though better coverage can be obtained under $[1, 1; 1, 0]$, the small cells still do not accomplish much as they experience severe intra and cross tier interference in the 800MHz band. 

\begin{figure}
\centering
\includegraphics[width=8cm]{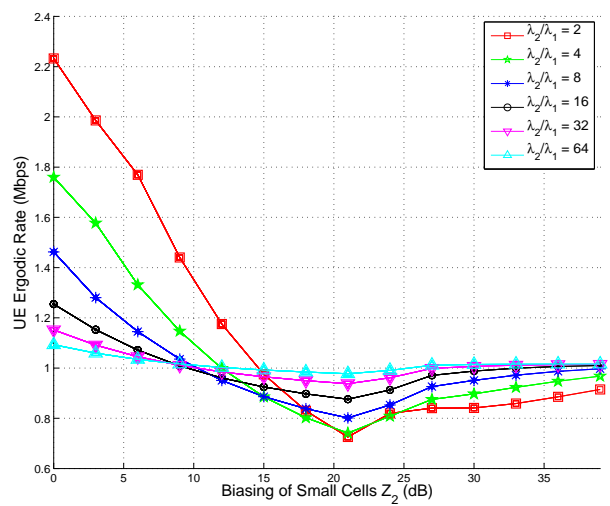}
\caption{UE ergodic rate vs. biasing under band deployment $[1, 1; 1, 0]$}
\label{fig:1110}
\end{figure}

Fig. \ref{fig:1011} and \ref{fig:1111} show the UE ergodic rates under band deployments $[1, 0; 1, 1]$ and $[1, 1; 1, 1]$, respectively. For similar reason as in the deployment $[1, 1; 1, 0]$, biasing is not helpful here; indeed, zero biasing value is optimal as shown in Fig. \ref{fig:1011} and \ref{fig:1111}. However, here we mainly take a sum-rate perspective; biasing may still be useful under $[1, 0; 0, 1]$ and $[1, 1; 0, 1]$ from other perspectives e.g. rate of cell-edge users. In addition, compared to the other 4 deployments, deployment $[1, 1; 1, 1]$ yields the largest UE ergodic rates in nearly all the scenarios of small cell densities, while $[1, 0; 1, 1]$ is almost as good as $[1, 1; 1, 1]$ except the low small cell density regime, i.e., $\lambda_2/\lambda_1 = 2$.

\begin{figure}
\centering
\includegraphics[width=8cm]{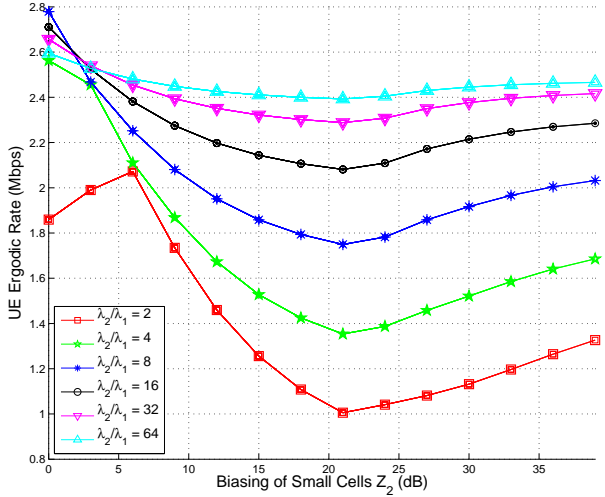}
\caption{UE ergodic rate vs. biasing under band deployment $[1, 0; 1, 1]$}
\label{fig:1011}
\end{figure}

\begin{figure}
\centering
\includegraphics[width=8cm]{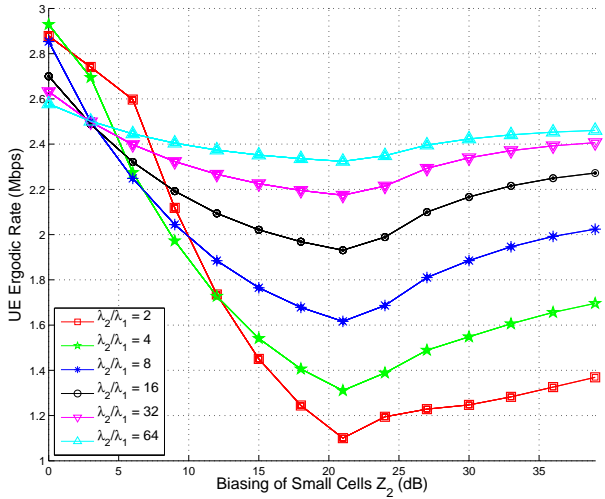}
\caption{UE ergodic rate vs. biasing under band deployment $[1, 1; 1, 1]$}
\label{fig:1111}
\end{figure}

Before ending this subsection, we remark that the above comparison of different band deployment scenarios is conducted purely from a sum-rate perspective, which is just one aspect of the network
design. In reality, deployment choices are made by jointly considering many other factors as well. In addition, interference management, which may affect the effect of biasing, is not incorporated.

\section{Conclusions}
\label{sec:conclusion}

This paper focuses on two central issues in CA-enabled HetNets: biasing and band deployment. To this end, we have proposed a general yet tractable $M$-band $K$-tier load-aware model. This new model provides a better characterization of the biasing effect and yields new design insights. Further, different band deployment configurations have also been studied and compared under the proposed model; the results reveal that universal cochannel deployment is the right goal to pursue. This work can be extended in a number of ways. In particular, many other distinct features such as multiple-input multiple-output (MIMO) technique, coordinated multi-point transmission/reception (CoMP) and device-to-device (D2D) communications can be jointly studied with CA.
 
\section*{Acknowledgment}

The authors thank Rapeepat Ratasuk and Bishwarup Mondal of Nokia Siemens Networks  for their help in identifying the problem and refining system model and numerical results. The authors also thank the anonymous reviewers for their valuable comments and suggestions, which helped the authors significantly improve the quality of the paper. 

\appendix

\subsection{Proof of Lemma \ref{lem:3}}
\label{proof:lem:3}

This proof follows the approach used in \cite{jo2011tractable}. In single flow CA, the typical UE scans over all the tiers and bands and then connects to the tier $k$ that provides the strongest biased received power in some band say $k^\star \in \mathcal{M}$. Then the UE performs CA with respect to this selected tier. Let $J$ be the random tier that the UE connects to. Then
by Assumption \ref{asum:2}, we have
\begin{align}
&\pi_k = \mathbb{P}(   J = k ) = \mathbb{P} ( Z_k P_{k} \|Y_{k_0} \|^{ - \alpha_{k^\star} } \notag \\
&= \max ( Z_\ell P_\ell \| Y_{\ell_0} \|^{ - \alpha_i } :    \ell \in \mathcal{K} , i \in \mathcal{M} )  ) \notag  \\
&= \mathbb{P} ( Z_k P_{k} \|Y_{k_0} \|^{ - \alpha_{k^\star} }  = \max ( Z_\ell P_{\ell} \| Y_{\ell_0} \|^{ - \alpha_{\ell^\star} } :    \ell \in \mathcal{K} )  )  \notag \\
&= \mathbb{P} ( Z_k P_{k} \|Y_{k_0} \|^{ - \alpha_{k^\star} }  \geq Z_\ell P_{\ell} \| Y_{\ell_0} \|^{ - \alpha_{\ell^\star} } , \forall \ell \neq k     )  \notag \\
&= \prod_{\ell \neq k} \b P ( \pi (Z_\ell P_\ell)^{ -\frac{2}{\alpha_{\ell^\star}} } \| Y_{\ell_0} \|^2 \geq \pi (Z_kP_k)^{ -\frac{2}{\alpha_{\ell^\star} } } \| Y_{k_0} \|^{\frac{2 \alpha_{k^\star}}{ \alpha_{\ell^\star} }}   ) , 
\label{eq:r1:9}
\end{align}
where  the last equality follows from the independence of $\{ \|Y_{\ell_0}\| \}$ (because the $K$ ground PPPs are independent). It is known that $\|Y_{\ell_0} \|$ is Rayleigh distributed (c.f. (\ref{eq:r1:8})) with pdf
$ f_{\|Y_{\ell_0} \| } ( r ) = 2 \pi \lambda_\ell r e^{ - \pi \lambda_\ell r^2 }, \ r \geq 0$. It follows that
\begin{align}
&\b P (  \pi  \|Y_{\ell_0} \|^{ 2}   (Z_\ell P_{\ell})^{- \frac{ 2 }{ \alpha_{\ell^\star}  } }  \geq x ) = \exp ( - (Z_\ell P_{\ell})^{ \frac{ 2 }{ \alpha_{\ell^\star}  } } \lambda_\ell  x ), \forall \ell \in \mathcal{K}. \notag
\end{align}
Using the above equality  and conditioning on $\| Y_{k_0} \| = r$ yields
\begin{align}
&\b P (  \pi  \|Y_{\ell_0} \|^{ 2}   (Z_\ell P_{\ell})^{- \frac{ 2 }{ \alpha_{\ell^\star}  } }  \geq \pi (Z_kP_k)^{ -\frac{2}{\alpha_{\ell^\star} } } r^{\frac{2 \alpha_{k^\star}}{ \alpha_{\ell^\star} }}  ) \notag  \\
 &= \exp( - \pi \lambda_\ell (\frac{ Z_\ell P_\ell }{Z_k P_k})^{ \frac{2}{\alpha_{\ell^\star}} } r^{\frac{2 \alpha_{k^\star}}{ \alpha_{\ell^\star} }}  ).
\label{eq:r1:11}
\end{align}
De-conditioning on $\| Y_{k_0} \| = r$ yields
\begin{align}
&P( \pi (Z_\ell P_\ell)^{ -\frac{2}{\alpha_{\ell^\star}} } \| Y_{\ell_0} \|^2 \geq \pi (Z_kP_k)^{ -\frac{2}{\alpha_{\ell^\star}} } \| Y_{k_0} \|^{\frac{2 \alpha_{k^\star}}{ \alpha_{\ell^\star} }}   )    \notag \\
&= \int_{0}^\infty \exp( -\pi \lambda_\ell (\frac{ Z_\ell P_\ell }{Z_k P_k})^{ \frac{2}{\alpha_{\ell^\star}} } r^{\frac{2 \alpha_{k^\star}}{ \alpha_{\ell^\star} }}  ) \cdot 2\pi \lambda_k r \exp( -\pi \lambda_k r^2 ) \dint r .
\label{eq:r1:10}
\end{align}
Plugging (\ref{eq:r1:10}) into (\ref{eq:r1:9}) yields that $\pi_k $ equals
\begin{align}
&\prod_{l\neq k} \int_{0}^\infty \exp( -\pi \lambda_\ell (\frac{ Z_\ell P_\ell }{Z_k P_k})^{ \frac{2}{\alpha_{\ell^\star}} } r^{\frac{2 \alpha_{k^\star}}{ \alpha_{\ell^\star} }}  ) \cdot 2\pi \lambda_k r \exp( -\pi \lambda_k r^2 ) \dint r \notag \\
&=  2\pi \lambda_k \int_{0}^\infty r \exp( -\pi \sum_{\ell \in \mathcal{K}} \lambda_\ell (\frac{ Z_\ell P_\ell }{Z_k P_k})^{ \frac{2}{\alpha_{\ell^\star}} } r^{\frac{2 \alpha_{k^\star}}{ \alpha_{\ell^\star} }}  )   \dint r. \notag 
\end{align}
This completes the proof.

\subsection{Proof of Lemma \ref{lem:4}}
\label{proof:lem:4}

Let $\Pi^{(i)}_k (x)$ denote the probability that the UE at position $x$ connects to tier $k$ in band $i$. Then the mean number of UEs attempting to connect to a BS of tier $k$ is given by
\begin{align}
L_{i,k} &= \frac{ \int_{\b R^2} \Pi^{(i)}_k (x) \Phi^{(u)} (d x) }{\int_{\b R^2}  \Phi_{k} (d x)} \notag  \\
&= \frac{ \int_{0}^{2\pi} \int_0^\infty \Pi^{(i)}_k (r,\theta) \lambda^{(u)} rdr d\theta  }{ \int_{0}^{2\pi} \int_0^\infty \lambda_k rdr d\theta} 
\notag  \\
& = \frac{\lambda^{(u)} \pi_k  \int_{0}^{2\pi} \int_0^\infty  rdr d\theta  }{\lambda_k \int_{0}^{2\pi} \int_0^\infty  rdr d\theta} 
 = 2 \pi \lambda^{(u)} G_k, \notag
\end{align}
where the third equality follows from the ergodicity of PPP $\Phi^{(u)}$, i.e., $\Pi^{(i)}_k (x) = \pi_k, \forall x$, and plugging $\pi^{(i)}_k$ yields the last equality. So if $x_{i,k}\neq 0 \textrm{ and } 0<    2 \pi \lambda^{(u)} G_k  \leq \frac{B_i}{b_{i,k}}$, tier $k$ is under-loaded in band $i$ and the load equals $b_{i,k} L_{i,k}/ B_i$. If $x_{i,k}\neq 0 \textrm{ and } 2 \pi \lambda^{(u)} G_k  > \frac{B_i}{b_{i,k}}$, tier $k$ is fully-loaded in band $i$ and the load equals $1$. If $x_{i,k} = 0$, it is trivial that $p_{i,k}=0$.

\subsection{Proof of Lemma \ref{lem:5}}
\label{proof:lem:5}
Note that 
$
P ( \| Y_{k_0} \|  \geq   r \mid J =k  ) = \frac{ P ( \| Y_{k_0} \|  \geq   r, J =k  )  }{\pi_k} 
$
where
\begin{align}
&\mathbb{P}( \|Y_{k_0} \| \geq r  ,  J = k )  \notag \\
&=\mathbb{P}( \|Y_{k_0} \| \geq r  \notag \\
& \quad \quad \quad \quad
\textrm{ and } Z_k P_{k} \cdot \| Y_{k_0} \|^{-\alpha_{k^\star}} \geq Z_\ell P_{\ell} \cdot  \| Y_{\ell_0} \|^{-\alpha_{\ell^\star}}, \forall \ell ) \notag \\
&= \int_r^{\infty} f_{\| Y_{k_0} \| } (x) \cdot \notag \\
& \quad \quad \b P (  \pi  \|Y_{\ell_0} \|^{ 2}   (Z_\ell P_{\ell})^{- \frac{ 2 }{ \alpha_{\ell^\star}  } }  \geq \pi (Z_kP_k)^{ -\frac{2}{\alpha_{\ell^\star} } } x^{\frac{2 \alpha_{k^\star}}{ \alpha_{\ell^\star} }} , \forall \ell  ) \dint x
\notag \\
&= \int_r^{\infty} f_{\| Y_{k_0} \| } (x) \cdot \notag \\
& \quad \quad \prod_{\ell \neq k} \b P (  \pi  \|Y_{\ell_0} \|^{ 2}   (Z_\ell P_{\ell})^{- \frac{ 2 }{ \alpha_{\ell^\star}  } }  \geq \pi (Z_kP_k)^{ -\frac{2}{\alpha_{\ell^\star} } } x^{\frac{2 \alpha_{k^\star}}{ \alpha_{\ell^\star} }}  ) \dint x \notag \\
&= \int_r^{\infty} f_{\| Y_{k_0} \| } (x) \prod_{\ell \neq k} \exp( - \pi \lambda_\ell (\frac{ Z_\ell P_\ell }{Z_k P_k})^{ \frac{2}{\alpha_{\ell^\star}} } r^{\frac{2 \alpha_{k^\star}}{ \alpha_{\ell^\star} }}  ) \dint x, \notag 
\end{align} 
where the last equality follows from (\ref{eq:r1:11}). Plugging $f_{\| Y_{k_0} \| } (x) = 2\pi \lambda_k r \exp( -\pi \lambda_k r^2 ), x \geq 0$, we get the ccdf of $\|Y_{k_0} \|$ conditional on $J=k$ as
\begin{align}
&\mathbb{P}( \|Y_{k_0} \| \geq r   \mid  J = k ) \notag \\
&= \frac{1}{\pi_k} 2 \pi \lambda_k \int_{r}^\infty x \cdot e^{ -\pi \sum_{\ell \in \mathcal{K}} \lambda_\ell (\frac{ Z_\ell P_\ell }{Z_k P_k})^{ \frac{2}{\alpha_{\ell^\star}} } x^{\frac{2 \alpha_{k^\star}}{ \alpha_{\ell^\star} }}  }  \dint x. \notag 
\end{align}
Differentiating $1 - \mathbb{P}( \|Y_{k_0} \| \geq r   \mid  J = k )$ with respect to $r$ yields the desired conditional pdf.

\bibliographystyle{IEEEtran}
\bibliography{IEEEabrv,first}


\end{document}